\documentclass[11pt]{article}
\usepackage[utf8x]{inputenc}
\usepackage{amsmath}
\usepackage{amssymb}
\usepackage{amsthm}
\usepackage{fullpage}
\usepackage[colorinlistoftodos]{todonotes}
\usepackage[bookmarks,colorlinks,breaklinks]{hyperref}  
\usepackage{upgreek}
\hypersetup{linkcolor=blue,citecolor=blue,filecolor=blue,urlcolor=blue}
\newcommand{\lref}[2][]{\hyperref[#2]{#1~\ref*{#2}}}
\renewcommand{\eqref}[1]{\hyperref[#1]{(\ref*{#1})}}

\numberwithin{equation}{section}

\usepackage{authblk}
\DeclareMathOperator{\Supp}{supp}
\DeclareMathOperator{\spa}{sparsity}
\DeclareMathOperator{\dimn}{dim}
\DeclareMathOperator{\vspan}{span}

\newcommand{\poly}{\mathrm{poly}}
\newcommand{\etal}{{\it et.~al.}}
\newcommand{\F}{{\mathbb F}}
\newcommand{\N}{{\mathbb N}}

\usepackage{algorithm}
\usepackage[noend]{algpseudocode}
\makeatletter
\def\BState{\State\hskip-\ALG@thistlm}
\makeatother

\newtheorem{theorem}{Theorem}[section]
\newtheorem{conjecture}[theorem]{Conjecture}

\newtheorem{lemma}[theorem]{Lemma}
\newtheorem{observation}[theorem]{Observation}
\newtheorem{proposition}[theorem]{Proposition}
\newtheorem{definition}[theorem]{Definition}
\newtheorem{claim}[theorem]{Claim}

\usepackage{enumitem}

\title{Sub-linear Upper Bounds on Fourier dimension of Boolean Functions in terms of Fourier sparsity}
\author{\rm Swagato Sanyal}
\affil{School of Technology and Computer Science \\ Tata Institute of Fundamental Research \\Mumbai 400005}
\affil{\textit{swagatos@tcs.tifr.res.in}}

\begin{document}

\maketitle

\begin{abstract}
  We prove that the Fourier dimension of any Boolean function with
  Fourier sparsity $s$ is at most $O\left(s^{2/3}\right)$. Our proof
  method yields an improved bound of $\widetilde{O}(\sqrt{s})$
  assuming a conjecture of Tsang~\etal~\cite{tsang}, that for every
  Boolean function of sparsity $s$ there is an affine subspace of
  $\mathbb{F}_2^n$ of co-dimension $O(\poly\log s)$ restricted to
  which the function is constant. This conjectured bound is tight upto
  poly-logarithmic factors as the Fourier dimension and sparsity of
  the address function are quadratically separated. We obtain
  these bounds by observing that the Fourier dimension of a Boolean
  function is equivalent to its non-adaptive parity decision tree
  complexity, and then bounding the latter.
\end{abstract}

\section{Introduction}

The study of Boolean functions involves studying various properties of
Boolean functions and their inter-relationships. Two such properties,
which we investigate in this article, are the Fourier dimension and
the Fourier sparsity, which were first studied in the context of property
testing by Gopalan~\etal~\cite{parikshit}. Given a Boolean function
$f:\F_2^n \to \{1,-1\}$ with Fourier expansion
\[f(x)=\sum_{\gamma \in \widehat{\mathbb{F}_2^n}}\widehat{f}(\gamma) \chi_\gamma(x),\]
Fourier dimension and Fourier sparsity are defined as follows.
\begin{definition}[Fourier dimension and sparsity]
 For a Boolean function $f:\F_2^n \to \{1,-1\}$ with Fourier expansion
\[f(x)=\sum_{\gamma \in \widehat{\mathbb{F}_2^n}}\widehat{f}(\gamma)
\chi_\gamma(x),\] the {\em Fourier support} of $f$, denoted by $\Supp(f)$,
is defined as 
\[\Supp(\widehat{f}) := \{\gamma \in \widehat{\mathbb{F}_2^n}: \widehat{f}(\gamma) \neq 0\}.\]
The {\em Fourier sparsity} of $f$, denoted by $\spa(f)$, is defined as the size of the support,
i.e., 
\[\spa(f) :=|\Supp(\widehat{f})|,\]
while the {\em Fourier dimension} $\dimn(f)$ of $f$ is defined as the
dimension of span of $\Supp(\widehat{f})$.
\end{definition}
The following inequalities easily follow from the definition of Fourier sparsity and dimension.
\begin{equation}
\label{eq:order}
 \log_2 \spa(f) \leq \dimn(f) \leq \spa(f).
\end{equation}
There are functions (e.g., indicator functions of subspaces) for which
the first inequality is tight. For the second inequality, the function
known to us having the closest gap between dimension and sparsity is
the \textit{address function} $Add_s:\{0,1\}^{\frac12 \log s +
  \sqrt{s}}\to \{0,1\}$, defined
as $$Add_s(x,y_1,y_2,\dots,y_{\sqrt{s}}) := y_x, \quad x \in
\{0,1\}^{\frac12\log s}, y_i \in \{0,1\}.$$ In other words, at any input
$(x,y)$, $Add_s(x,y)$ is the value of the addresee input bit $y_x$
indexed by the addressing variables $x$. The address function\footnote{To be precise, we should consider 
the $\pm 1$ version of the address function described here, where the $0$ and $1$ in the range
are interpreted as $+1$ and $-1$ respectively.} has
sparsity $s$ and dimension at least $\sqrt{s}$.
It is believed
that this is the tight upper bound for $\dimn(f)$ in terms of
$\spa(f)$. I. e., it is believed that the upper bound in
\eqref{eq:order} can be improved to $\dimn(f) \leq
\sqrt{\spa(f)}$\footnote{This is one of the conjectures proposed in
  the open problem session at the Simons workshop on Real Analysis in Testing, Learning and Inapproximability.}. 

Our main result is the following, which to our knowledge is the first improvement over
the trivial $\dim(f) \leq \spa(f)$ bound.
\begin{theorem}
 \label{thm:th1}
Let $f$ be a Boolean function with $\spa(f)=s$. Then, $\dim(f)=O\left(s^{2/3}\right).$
\end{theorem}

This result is proved using a lemma of
Tsang~\etal~\cite{tsang} bounding the co-dimension of an affine subspace
restricted to which the function reduces to a constant, in terms of
Fourier sparsity of the function. 

\begin{lemma}[Corollary of {\cite[Lemma~30]{tsang}}]
\label{lem:set}
Let $f:\mathbb{F}_2^n \rightarrow \{1,-1\}$ be a Boolean function
with Fourier sparsity $s$. Then there is an affine subspace $V$ of $\mathbb{F}_2^n$ of co-dimension $O(\sqrt{s})$ such
that $f$ is constant on $V$.
\end{lemma}

Tsang~\etal~\cite{tsang} proved a more general result in terms of
Fourier $l_1$-norm (see \lref[Section]{sec:prelim} for more
details). Tsang~\etal~proved this result while trying to investigate
the {\em log rank conjecture} in communication complexity for {\em
  xor} functions. The log rank conjecture is a long standing and
important conjecture in communication complexity. The statement of the
conjecture is that the deterministic communication complexity of a
Boolean function is asymptotically bounded above by some fixed
poly-logarithm of the rank of it's communication
matrix. Tsang~\etal~\cite{tsang} suggested a direction towards proving
log-rank conjecture for an important class of functions called xor
functions. A Boolean function $f(x,y)$ on two $n$ bit inputs is a xor
function if there exists a Boolean function $F$ on $n$ bits such that
$f(x,y)=F(x \oplus y)$.  In particular, they propose a protocol for
such a $f$ based on the parity decision tree of $f$ and show that the
communication complexity of this proposed protocol is polylogarithmic
in rank of the communication matrix if the following related conjecture is true.
\begin{conjecture}[{\cite[Conjecture~27]{tsang}}]
 \label{conj:tsang}
 There exists a constant $c > 0$ such that for every Boolean function
 $f$ with Fourier sparsity $s$, there exists an affine subspace of
 co-dimension $O\left(\log^c s\right)$ on which $f$ is constant.
\end{conjecture}
Tsang~\etal~prove the above conjecture for certain classes of
functions, which include functions with constant $\mathbb{F}_2$
degree and prove \lref[Lemma]{lem:set} for general functions. Our next
result shows that if we assume this conjecture instead of
\lref[Lemma]{lem:set}, we can improve the bound in
\lref[Theorem]{thm:th1} to the following (which is optimal upto
poly-logarithmic factors).
\begin{theorem}
 \label{thm:th2}
Let $f$ be a Boolean function with Fourier sparsity $s$.
Assuming \lref[Conjecture]{conj:tsang}, $\dim(f)=\widetilde{O}\left(\sqrt{s}\right).$
\end{theorem}

\paragraph{Proof Idea:} We begin by making a simple, but crucial
observation that the Fourier dimension of a Boolean function is
equivalent to its non-adaptive parity decision tree complexity (see
\lref[Proposition]{prop:basic}). This offers us a potential approach
towards upper bounding the Fourier dimension of a Boolean function:
exhibiting a shallow non-adaptive parity decision tree of the
function. 

Towards this end, we first recall the construction of the
(adaptive) parity decision tree of Tsang~\etal~\cite{tsang}, which in
turn improves on an earlier construction due to 
Shpilka~\etal~\cite[Theorem~1.1]{shpilka}. The broad idea of their
construction is as follows: At any point in time, a partial
tree is maintained whose leaves are functions which are restrictions of the original
function on different affine subspaces. Then a non-constant
leaf is picked arbitrarily, and a small set
of linear restrictions is obtained by invoking \lref[Lemma]{lem:set}, such that the restricted function at that leaf becomes
constant. The next step is observing that if the function at the same leaf is restricted to
all the affine subspaces obtained by setting the same set of 
parities in all possible ways, the sparsity of each of the corresponding restricted functions
is at most half of that of the original function. This is because, in
the former restriction, since the function becomes constant, the
Fourier coefficients corresponding to non-constant characters must
disappear in the restricted space. This can only happen if every
non-constant parity gets identified with at least one
other parity. This identification leads to halving of the
support. Proceeding in this way, they obtain a parity decision tree
of depth O($\sqrt{s}$).

Note that the choice of parities depends on the leaf (function)
chosen, and hence on the outcomes of the preceding queries. Thus the
constructed tree is an adaptive one. In this article, we make this
tree non-adaptive, at the cost of a small increase in depth. At each
level, we choose an appropriate function, invoke
\lref[Lemma]{lem:set}, and obtain the restrictions which make it
constant. Then we query the same set of parities at every
leaf. The next step is arguing that this leads to a significant
reduction of sparsity in the next level. This is done using the
Uncertainty Principle (\lref[Theorem]{thm:up}). Continuing in
this fashion, we show that in a small number of levels, the size of the
union of the Fourier supports of all the leaves becomes so small that
we can query all of them, thereby turning all the leaves into
constants.

\section{Preliminaries}\label{sec:prelim}

let $f:\F_2^n \rightarrow \{1,-1\}$ be a Boolean function. We
think of the range $\{+1,-1\}$ as a subset of $\mathbb{R}$. The inputs
to $f$ are $n$ variables $x_1, \ldots, x_n$ which take values in
$\F_2$. We identify the additive group in $\mathbb{F}_2$
with the group $\{+1,-1\}$ under real number multiplication, and think
of the variables as taking $+1$ and $-1$ values, where $0$ and $1$ of
$\F_2$ get mapped to $+1$ and $-1$ respectively. We denote
this group isomorphism by $(-1)^{(\cdot)}$, \i.e., $(-1)^0$ is $1$ and
$(-1)^1$ is $-1$. When the $x_i$'s are $\pm 1$, it is well known that
every Boolean function $f(x)$ (where $x$ stands for $x_1,\ldots, x_n$)
can be uniquely written as
\[f(x)=\sum_{S \subseteq [n]}\widehat{f}(S)\prod_{i \in S}x_i.\] Thus,
when the variables are $\pm 1$, $f$ can be written as a multilinear
real polynomial. For every $S \subseteq [n]$, the product $\prod_{i \in
  S}x_i$ is the logical xor of the bits in $S$, and $\widehat{f}(S)$ is a
real number.  These products are exactly the \emph{characters} of
$\mathbb{F}_2^n$, which are $\pm 1$ versions of the linear forms
belonging to the dual vector space $\widehat{\mathbb{F}_2^n}$ of
$\mathbb{F}_2^n$.  We adopt the following notation in this paper:
\[f(x)=\sum_{\gamma \in \widehat{\mathbb{F}_2^n}}\widehat{f}(\gamma) \chi_\gamma(x).\]
Here, each $\gamma \in \widehat{\mathbb{F}_2^n}$ is a linear function
from $\mathbb{F}_2^n$ to $\mathbb{F}_2$, and $\chi_\gamma$ is $(-1)^\gamma$.

We recall some standard definitions and facts about the Fourier
coefficients.
\begin{definition}
Let $f(x)=\sum_{\gamma \in \widehat{\mathbb{F}_2^n}}\widehat{f}(\gamma) \chi_\gamma(x)$ be a Boolean function. The $p$-th spectral norm $\|\widehat{f}\|_p$ of $f$ is defined as:
\[\|\widehat{f}\|_p := \left[\sum_{\gamma \in \widehat{\F_2^n}}|\widehat{f}(\gamma)|^p\right]^{1/p}.\] 
\end{definition}

\begin{lemma}[Parseval's identity]
 For a Boolean function $f$, $\|\widehat{f}\|_2 = 1$.
\end{lemma}
The 1st spectral norm of a Boolean function can be bounded via sparsity as follows.
\begin{claim}\label{claim:lonespar}For a Boolean function $f$, $\|\widehat{f}\|_1 \leq \sqrt{s}$.
\end{claim}
\begin{proof}
\begin{align*}
\|\widehat{f}\|_1  \leq \|\widehat{f}\|_2. \sqrt{s}   = \sqrt{s}.
\end{align*}
The first inequality follows due to Cauchy-Schwarz inequality while the
second equality follows from Parseval's identity.
\end{proof}

For proving our results, we shall use the following version of the Uncertainty Principle. The reader is referred to \cite{aobf} for a proof.
\begin{theorem}[Uncertainty Principle]
\label{thm:up}
 Let $p: \mathbb{R}^n \rightarrow \mathbb{R}$ be a real multilinear
 $n$-variate polynomial with sparsity $s$ (i.e, it has $s$ monomials with non-zero coefficients). Let $U_n$ denote the 
uniform distribution on $\{+1,-1\}^n$. Then
\[\Pr_{x \sim U_n}[p(x) \neq 0] \geq \frac{1}{s}.\]
\end{theorem}
As stated in the introduction, we need the following theorem due to
Tsang~\etal~\cite{tsang}. 
\begin{theorem}[{\cite[Lemma 30]{tsang}}]
\label{sphilka}
let $f:\mathbb{F}_2^n \rightarrow \{1,-1\}$ be such that $\|\widehat{f}\|_1=A$. Then there is an affine subspace $V$ of $\mathbb{F}_2^n$ of co-dimension $O(A)$ such
that $f$ is constant on $V$.
\end{theorem}

\lref[Lemma]{lem:set} is a simple corollary of this theorem via \lref[Claim]{claim:lonespar}.

We end this section by a simple proposition which is crucial to our proofs.
\begin{definition}[non-adaptive parity decision tree complexity]
Let $f$ be a Boolean function. The \emph{non-adaptive parity decision
  tree complexity} of $f$, (denoted by $\text{NADT}_\oplus(f)$), is defined as the minimum integer $t$ such that there
exist $t$ linear forms $\gamma_1, \ldots, \gamma_t \in \widehat{\mathbb{F}_2^n}$ such that $f$ is a junta of $\gamma_1, \ldots, \gamma_t$. In other words, on every input, specifying the outputs of the $\gamma_i$'s
specifies the output of $f$.
\end{definition}
\begin{proposition}
\label{prop:basic}
 For a Boolean function $f$, $\mbox{NADT}_\oplus(f)$ = dim($f$).
\end{proposition}
\begin{proof}
If the outputs of a basis of span of $\Supp(\widehat{f})$ is specified, then that clearly specifies the outputs of all characters in $\Supp(\widehat{f})$, and hence it specifies the output of the function. Thus $\mbox{NADT}_\oplus(f) \leq$ dim($f$). \\ \\
Now, Let $\mbox{NADT}_\oplus(f)=t$. Let the outputs of $\gamma_1, \ldots, \gamma_t$ specify the output of $f$, and without loss of generality assume these linear forms to be linearly independent as vectors in $\widehat{\mathbb{F}_2^n}$.
Arbitrarily extend $\gamma_1, \ldots, \gamma_t$ to a basis $\gamma_1, \ldots, \gamma_n$ of $\widehat{\mathbb{F}_2^n}$.
For $x=(x_1, \ldots, x_n) \in \mathbb{F}_2^n$, let $L(x)=(\gamma_1(x), \ldots, \gamma_n(x))$. $L$ is easily seen to be an invertible linear transformation from $\mathbb{F}_2^n$
onto itself. Now, $\forall x \in \mathbb{F}_2^n, \forall i=1,\ldots, n, \gamma_i(x)=(L(x))_i$. Replacing $x$ by $L^{-1}(x)$ we have $\gamma_i(L^{-1}(x))=x_i$. Now consider the Boolean function
$g(x)=f(L^{-1}(x))=\sum_{\gamma \in \widehat{\mathbb{F}_2^n}}\widehat{f}(\gamma) (-1)^{\gamma(L^{-1}(x))}$. Clearly $\dimn(g)=\dimn(f)$. Also, $g$ is completely specified by the outputs
of $\gamma_i(L^{-1}(x))$'s for $i = 1, \ldots, t$. Since $\gamma_i(L^{-1}(x))=x_i$, we have that $g$ is a junta of $x_1, \ldots, x_t$. Thus all the monomials in $\Supp(\widehat{g})$ contain only the variables $x_1, \ldots, x_t$. Thus $\dimn(f) = \dimn(g)
\leq t = \mbox{NADT}_\oplus(f)$. \\ \\
The proposition follows by combining the two inequalities.
\end{proof}

\section{Upper Bounding Parity Decision Tree Complexity}

In this section, we upper bound the non-adaptive parity decision tree
complexity of a Boolean function $f$ with Fourier sparsity at most
$s$. Consider the following procedure, parametrized by a parameter
$\tau \in \N$ (that we will set later) that constructs the non-adaptive
parity decision tree.

\newcommand{\nadtproc}{\text{\sc Non-adaptive-parity-decision-tree-procedure}}
\begin{itemize}
\item[] $\nadtproc_\tau(f)$
\item[] Input: Boolean function $f: \F_2^n \to \{1,-1\}$; Parameter: $\tau
  \in \N$
\begin{enumerate}
\item Set $\Gamma \gets \emptyset$, $\mathcal{S} \gets\Supp(\widehat{f})$
  and 
  $\mathcal{F} \gets \{f\}$.
\item\label{step:while} While $|\mathcal{S}| > \tau$, do
\begin{enumerate}
\item\label{step:chooseg} Let $g$ be a function in $\mathcal{F}$ with the largest Fourier
  sparsity. Let $\gamma_1,\ldots,\gamma_{n_g}$ be linear functions and
  $b_1, \ldots, b_{n_g} \in \mathbb{F}_2$ be such that a largest
  affine subspace on which $g$ is constant is $\{x \in \F_2^n :
  \gamma_1(x)=b_1, \ldots, \gamma_{n_g}(x) = b_{n_g}\}$.  Query
  $\gamma_1,\ldots, \gamma_{n_g}$.
\item Set $\Gamma \gets \Gamma \cup \{\gamma_1,\ldots,
  \gamma_{n_g}\}$.
\item For each $b=(b_\gamma)_{\gamma \in \Gamma} \in
  \mathbb{F}_2^{|\Gamma|}$, let $V_b$ be the affine subspace $\{x \in \F_2^n:
  \forall \gamma \in \Gamma, \gamma(x) = b_\gamma\}$.  Set
  $\mathcal{F} \gets \bigcup_{b \in \mathbb{F}_2^{|\Gamma|}} \{f |_{V_b}\}$.
\item $\mathcal{S} \gets \bigcup_{h \in \mathcal{F}} \Supp(\widehat{h})$.
\end{enumerate}
\item Query all the parities in $\mathcal{S}$.
\end{enumerate}
\end{itemize}

Notation: After each iteration of the while loop in the procedure,
$\Gamma$ is the set of parities that have been queried so far,
$\mathcal{F}$ is the set of all restrictions of $f$ to the affine
subspaces obtained by different assignments to parities in $\Gamma$
and $\mathcal{S}$ the union of the Fourier supports of functions in
$\mathcal{F}$. Let  $\Gamma^{(i)}, \mathcal{F}^{(i)}$ and
$\mathcal{S}^{(i)}$ denote $\Gamma, \mathcal{F}$ and $\mathcal{S}$
resepectively at
the end of the $i$-th iteration of the while loop.

For each $i$, let $b=(b_\gamma)_{\gamma \in \Gamma^{(i)}} \in
\mathcal{F}_2^{|\Gamma^{(i)}|}$ and let $V_{b}$ be the affine subspace
defined by linear constraints $\{\gamma(x)=b_\gamma: \gamma \in
\Gamma^{(i)}\}$.  In $V_b$, more than one linear functions of the
original space may get identified as same.\footnote{By \textquoteleft
  same\textquoteright\ we also include their being negations of each
  other as the smaller subspace is an affine space and not always a
  vector space.}  More specifically, $\delta_1$ and $\delta_2$ get
identified as same in $V_b$ if and only if $\delta_1 + \delta_2 \in
\vspan\Gamma^{(i)}$. Thus, $\Supp(\widehat{f})$ gets partitioned
into equivalence classes, such that for each class, for every $b \in
\mathcal{F}_2^{|\Gamma^{(i)}|}$, the linear functions belonging to
that class are identified as same in $V_b$. 

Let $l^{(i)}$ denote the number of cosets of the subspace
$\vspan\Gamma^{(i)}$ with which $\Supp(\widehat{f})$ has non-empty
intersection. For $j=1,\ldots,l^{(i)}$, let $\beta_j^{(i)}$ be some
representative element in $\Supp(\widehat{f})$ of the $j$-th coset of
$\vspan \Gamma^{(i)}$ having non-empty intersection with
$\Supp(\widehat{f})$. For each $j$, let $\beta_j^{(i)} + \alpha_{j,1}^{(i)},
\ldots, \beta_j^{(i)} + \alpha_{j,{k_j}}^{(i)}$ be the $k_j^{(i)} (\geq
1)$ elements in $\Supp(\widehat{f})$ which are in the same coset of
$\mbox{span\ } \Gamma^{(i)}$ as $\beta_j^{(i)}$. For each $i,j$,
define the polynomials $P_j^{(i)}(x) := \displaystyle\sum_{l=1}^{k_j}
\widehat{f}\left(\beta_j^{(i)}+\alpha_{j,l}^{(i)}\right)\chi_{\alpha_{j,l}^{(i)}}(x)$. Note that the polynomials $P_j^{(i)}$, $j=1,\ldots,l^{(i)}$, are non-zero.

Given this notation, we can then write the Fourier expansion of $f$ in
the following form:
\[f(x)=\displaystyle\sum_{j=1}^{l^{(i)}}P_j^{(i)}(x)\chi_{\beta_j^{(i)}}(x).\]

\begin{observation}
\label{obs:sum}
$\forall i, \displaystyle\sum_{j=1}^{l^{(i)}}k_j^{(i)}=s$.
\end{observation}
\begin{observation}
\label{obs:lbound}
$|\mathcal{S}^{(i)}| = l^{(i)}$.
\end{observation}
We now argue that after every iteration of the while loop, there
exists a function $h \in \mathcal{F}^{(i)}$ which has large support.
\begin{lemma}
\label{lem:support}
After $i$-th iteration, there exists a $h \in \mathcal{F}^{(i)}$ such
that $|\Supp(\widehat{h})|$ is at least $\left({l^{(i)}}\right)^2/{s}$.
\end{lemma}
\begin{proof}
Consider any function $f|_{V_b} \in \mathcal{F}^{(i)}$. The Fourier
decomposition of $f|_{V_b}$ is given by $f|_{V_b} =
\sum_{j=1}^{l^{(i)}} P_j^{(i)}(b)\chi_{\beta_j^{(i)}}(x)$. Thus,
$|\Supp(\widehat{f|_{V_b}})|$ is exactly the number of polynomials
$P_j^{(i)}, j=1,\dots,l^{(i)}$ such that $P_j^{(i)}(b)$ is non-zero. We analyze this
quantity as follows. Pick a $b \in \mathbb{F}_2^{|\Gamma^{(i)}|}$ uniformly at random. For
each $j$, $j=1,\ldots,l^{(i)}$, by \lref[Theorem]{thm:up},
$\Pr_b[P_j^{(i)}(b) \neq 0] \geq \frac{1}{k_j^{(i)}}$ (since
each $P_j^{(i)}$ is a non-zero polynomial). Thus,
\begin{align*}
\mathbb{E}_b\left[|\Supp(\widehat{f|_{V_b}})|\right] &\geq \sum_{j=1}^{l^{(i)}}
\frac{1}{k_j^{(i)}} 
\geq l^{(i)} \cdot \frac{1}{\left(\sum_{j=1}^{l^{(i)}}
    k_j^{(i)}\right)/l^{(i)}}  && \text{[By convexity of $1/x$]}\\
&= \frac{\left(l^{(i)}\right)^2}{s} && \text{[By
  \lref[Observation]{obs:sum}]}.
\end{align*}
Hence, there exists a $h \in \mathcal{F}^{(i)}$ such
that $|\Supp(\widehat{h})|$ is at least $\left({l^{(i)}}\right)^2/{s}$.
\end{proof}

\begin{lemma}
\label{lem:reduction}
 Assume that $\nadtproc_\tau(f)$  runs for $t$ iterations. Then for all $i$, $i=1,\ldots, t-1$,
\[l^{(i+1)} \leq l^{(i)}-\frac{\left({l^{(i)}}\right)^2/s-1}{2}.\]
\end{lemma}
\begin{proof}
  Let $g$ be the chosen function at \lref[Step]{step:chooseg} in the
  $(i+1)$-th iteration of the procedure. Let
  $\gamma_1,\ldots,\gamma_{n_g}$ be the parities queried at that
  step. Hence there is $b=(b_1,\ldots,b_{n_g}) \in \mathbb{F}_2^{n_g}$
  such that $g$ is constant on the affine subspace $V_b$ obtained by
  setting each $\gamma_j$ to $b_j$ for $j=1,\ldots,n_g$. Since $g$ is
  constant on $V_b$, each non-zero parity in it's Fourier support must
  disappear in $V_b$. Thus, for every $b' =(b')_j \in
  \mathbb{F}_2^{n_g}$, in the affine space $V_{b'}$ obtained by
  restricting each $\gamma_j$ to $b'_j$, every non-zero parity in
  $\Supp(\widehat{g})$ is matched to some other parity in
  $\Supp(\widehat{g})$. Since $\Supp(\widehat{g}) \subseteq
  \mathcal{S}^{(i)}$, it follows that $|\mathcal{S}^{(i+1)}|$ is at
  least $\frac{\mbox{$|\Supp(\widehat{g})|$}-1}{2}$ less than
  $|\mathcal{S}^{(i)}|$. The proof now follows from
  \lref[Lemma]{lem:support} and \lref[Observation]{obs:lbound}.
\end{proof}
\begin{lemma}
 \label{lem:final}
Let $\nadtproc$ be run with parameter $\tau \geq \sqrt{2s}$. Assume that it runs for $t$ iterations. Then for $i=1,\ldots, t$, $l^{(i)} \leq \frac{4s}{i}$.
\end{lemma}
\begin{proof}
 We will prove it by induction on $i$. Base case, $i=1$, is trivial as
 $l^{(1)}$ can be at most $s$. 

Now let us assume that the statement is true for all $i \leq m$. From
\lref[Lemma]{lem:reduction}, we have that $l^{(m+1)} \leq
l^{(m)}-\left(\left({l^{(m)}}\right)^2/s-1\right)/{2}$. Since $\gamma \geq
\sqrt{2s}$, $(\left({l^{(m)}}\right)^2/s-1)/{2}$ can be
lower bounded by $\left(l^{(m)}\right)^2/4s$. We thus have 
\begin{align*} l^{(m+1)} &\leq l^{(m)}-\frac{\left({l^{(m)}}\right)^2}{4s}
 =s-\left(\sqrt{s}-\frac{l^{(m)}}{2\sqrt{s}}\right)^2 \\
& \leq s-\left(\sqrt{s}-\frac{2s}{m\sqrt{s}}\right)^2 & \text{[By inductive hypothesis]} \\
& =\frac{4s(m-1)}{m^2} = \left(\frac{4s}{m+1}\right) \left( \frac{m^2-1}{m^2} \right) \\
& \leq \frac{4s}{m+1}. 
\end{align*}
This completes the proof of the lemma.
\end{proof}
Theorems~\ref{thm:th1} and \ref{thm:th2} follow easily from the above
lemma using \lref[Lemma]{lem:set} and \lref[Conjecture]{conj:tsang}
respectively as follows.
\begin{proof}[Proof of {\lref[Theorem]{thm:th1}}]
  Run the \nadtproc \ with parameter $\tau =
  \Theta\left(s^{2/3}\right)$. Since $l^{(i)} \leq \frac{4s}{i}$
  (\lref[Lemma]{lem:final}), the procedure terminates after
  $t=O(s^{1/3})$ iterations. From \lref[Lemma]{lem:set}, in the
  $i$-th iteration, the number of parities set is at most
  $O(\sqrt{l^{(i)}})$. Thus the total number of queries made by the
  procedure is
\begin{align*}
\displaystyle\sum_{i=1}^t O\left(\sqrt{l^{(i)}}\right) + \tau
 &= \displaystyle\sum_{i=1}^t O\left(\sqrt{\frac{s}{i}}\right) + \tau
 &&\text{[By \lref[Lemma]{lem:final}]}\\
 & =\sqrt{s}\left( \displaystyle\sum_{i=1}^t O\left(\sqrt{\frac{1}{i}}\right)\right) + \tau &&=\sqrt{s} \left( O \left( \int_1^t \frac{dx}{\sqrt{x}}\right)\right) + \tau 
 =O\left(\sqrt{s}. \sqrt{t}\right) + \tau \\
 & =O\left(s^{1/2 + 1/6} \right) + \tau && \text{[Since $t=O\left(s^{1/3})\right)$]}\\ 
 &=O\left(s^{2/3}\right) && \text{[Since $\tau=\Theta\left(s^{2/3}\right)$]}.
\end{align*}
Thus, $\mbox{NADT}_\oplus(f) = O\left(s^{2/3}\right)$. From \lref[Proposition]{prop:basic}, it follows that dim($f$)=$\mbox{NADT}_\oplus(f) = O\left(s^{2/3}\right)$.
\end{proof}

\begin{proof}[Proof of {\lref[Theorem]{thm:th2}}]
  Run the \nadtproc \  with parameter $\tau=2\sqrt{s}$. By
  \lref[Lemma]{lem:final}, it runs for at most
  $\frac{4s}{2\sqrt{s}}=2\sqrt{s}$ iterations. If
  \lref[Conjecture]{conj:tsang} is true, the total number of parities
  set is at most $O\left(\log^c s\right)2\sqrt{s} + \tau$ which is
  $\widetilde{O}\left(\sqrt{s}\right)$.
\end{proof}
\section{Acknowledgements}
The author would like to thank Arkadev Chattopadhyay and Prahladh Harsha for many helpful discussions. The author is thankful to Prahladh Harsha for his help in improving the presentation of this article significantly.

%
%
%

\bibliographystyle{alpha}
\bibliography{dimspar}

\end{document}